\documentclass[12pt]{article}
\usepackage[mathscr]{eucal}
\usepackage{amssymb,latexsym}
\usepackage{verbatim}
\usepackage{graphicx}
\usepackage{amsmath}
\usepackage{amsthm}
\usepackage{enumerate}
\usepackage{authblk}

\usepackage{xcolor}

\usepackage{appendix}

\usepackage[normalem]{ulem}

\newtheorem{thm}{Theorem}[section]
\newtheorem{lem}[thm]{Lemma}
\newtheorem{Def}{Definition}[section]
\newtheorem{prop}[thm]{Proposition}
\newtheorem{cor}[thm]{Corollary}

\newcommand\calD{{\mathcal{D}}}

\newcommand\calK{{\mathcal{K}}}

\newcommand\bbR{{\mathbb R}}
\newcommand\bbZ{{\mathbb{Z}}}

\newcommand\ol{\overline}

\renewcommand\S{\Sigma}

\renewcommand\d{\partial}

\renewcommand\b{\beta}

\newcommand\ric{{\rm Ric}}

\newcommand\g{\gamma}

\renewcommand\a{\alpha}

\newcommand\beq{\begin{equation}}
\newcommand\eeq{\end{equation}}
\newcommand\ben{\begin{enumerate}}
\newcommand\een{\end{enumerate}}
\newcommand\bit{\begin{itemize}}
\newcommand\eit{\end{itemize}}

\newcommand{\R}{\mathbb R}

\newcommand{\J}{\mathcal{J}}

\newcommand{\ov}{\overline}

\newcommand{\pd}{\partial}

\DeclareFontFamily{OT1}{rsfs}{} \DeclareFontShape{OT1}{rsfs}{m}{n}{ <-7> rsfs5 <7-10> rsfs7 <10-> rsfs10}{}
\DeclareMathAlphabet{\mycal}{OT1}{rsfs}{m}{n}

\DeclareFontFamily{OT1}{rsfs}{} \DeclareFontShape{OT1}{rsfs}{m}{n}{ <-7> rsfs5 <7-10> rsfs7 <10-> rsfs10}{}
\DeclareMathAlphabet{\mycal}{OT1}{rsfs}{m}{n}

\newcounter{mnotecount}

\title{A conformal infinity approach to asymptotically $\text{AdS}_2\times S^{n-1}$ spacetimes}

\author[1]{Gregory J. Galloway}
\author[2]{Melanie Graf\,}
\author[3]{Eric Ling}

\affil[1]{University of Miami
}

\affil[2]{University of Washington}

\affil[3]{KTH, Stockholm}

\begin{document}
\date{}
\maketitle

\vspace{.1in}

\begin{abstract} 
It is well known that the spacetime  $\text{AdS}_2\times S^2$ arises as the `near horizon' geometry of the extremal Reisser-Nordstrom solution, and for that reason it has been studied in connection with the AdS/CFT correspondence. Motivated by a conjectural viewpoint of Juan Maldacena, the authors in \cite{GalGraf} studied the rigidity of asymptotically $\text{AdS}_2\times S^2$ spacetimes satisfying the null energy condition. In this paper, we take an entirely different and more general approach to the asymptotics based on the notion of conformal infinity. This involves a natural modification of the usual notion of timelike conformal infinity for asymptotically anti-de Sitter spacetimes.  As a consequence we are able to obtain a variety of new results, including similar results to those in \cite{GalGraf} (but now allowing both higher dimensions and more than two ends) and a version of topological censorship.
\vspace{.2in}

\end{abstract}

\tableofcontents

\section{Introduction}

It is a well known fact that the spacetime ${\rm AdS}_2 \times S^2$ arises as the `near horizon' geometry of the extremal Reissner-Nordstrom solution.  In fact, ${\rm AdS}_2 \times S^2$,  and also ${\rm AdS}_2 \times S^3$,  appear 
in uniqueness results for near-horizon supersymmetric solutions of minimal supergravity, and in 
near-horizon solutions in Einstein-Maxwell-Chern-Simons theory in four and five dimensions;  see especially the review article \cite{KL13}.    
${\rm AdS}_2 \times S^2$ has also been discussed in the context of string theory and the AdS/CFT correspondence, as for example in \cite{MMS}.  Based on certain examples in \cite{MMS}, and other considerations, Maldacena \cite{MaldPC} recently suggested that spacetimes which satisfy the null energy condition (or average null energy condition), and which  asymptote to ${\rm AdS}_2 \times S^2$ at infinity should be quite `rigid'.  
Following this suggestion, in \cite{GalGraf}, the authors studied the rigidity of asymptotically ${\rm AdS}_2 \times S^2$ spacetimes satisfying the null energy condition.  

Among the results obtained, we showed that such spacetimes admit two transverse foliations by totally geodesic null hypersurfaces, each extending from one end to the other, the intersections of which give rise to a foliation of spacetime by totally geodesic round $2$-spheres.   These are standard features of  ${\rm AdS}_2 \times S^2$. However, without imposing some stronger condition, we were unable to 
conclude that such a spacetime is actually isometric to ${\rm AdS}_2 \times S^2$.  In fact, using the Newman-Penrose formalism, Tod \cite{Tod1, Tod2} constructed examples of asymptotically ${\rm AdS}_2 \times S^2$ spacetimes satisfying the null energy condition, having the structural properties established in \cite{GalGraf}, which are not isometric to ${\rm AdS}_2 \times S^2$. Some further examples will be given in the present paper.  However, in the presence of  certain field equations, the only possibility seems to be ${\rm AdS}_2 \times S^2$.

The approach to the asymptotics taken in \cite{GalGraf} was to require that on each of two ``external" spacetime regions, the spacetime metric $g$ asymptotes at a precise rate, with respect to a well chosen coordinate system,   to the $AdS_2 \times S^2$ metric, on approach to infinity.  These asymptotic conditions were also supplemented by certain causal theoretic conditions on the complement of  the external regions.  The asymptotic analysis in \cite{GalGraf}, while fairly technical, gave very precise control over the causal and geometric properties of asymptotically ${\rm AdS}_2 \times S^2$ spacetimes, defined in this manner.  

In this paper, we take an entirely different and more general approach to the asymptotics, based on the notion of conformal infinity.  In the next section we define what it means for a spacetime to have an `asymptotically ${\rm AdS}_2 \times S^{n-1}$ end', in terms of it admitting a certain type of timelike conformal boundary.  The approach is based on the simple observation, spelled out in the next section, that ${\rm AdS}_2 \times S^{n-1}$ conformally embeds in a natural way into the Einstein static universe, ($\bbR \times S^n, -dt^2 + d\omega_n^2)$.   Via this embedding, ${\rm AdS}_2 \times S^{n-1}$ acquires two topological boundary components, namely  two $t$-lines, through, say, the north and south pole of $S^n$. 
Thus, as made precise in Definition \ref{maindef}, a spacetime $(M^{n+1},g)$ has an asymptotically ${\rm AdS}_2 \times S^{n-1}$ end if it conformally embeds into a globally hyperbolic spacetime $(\ov{M}^{n+1},\ov{g})$, in which the conformal boundary consists of a smooth  inextendible timelike curve $\J$.  This situation, of course differs, from the standard definition of the conformal boundary of an asymptotically anti-de Sitter spacetime, in which the conformal boundary is a timelike hypersurface.  Moreover, in our definition, the conformal factor will not be smooth at $\J$, but rather will only be continuous (in fact Lipschitz) there.  Such is the case with the embedding of ${\rm AdS}_2 \times S^{n-1}$ into the Einstein static universe. 

As a consequence of the approach taken here, we are able to obtain a variety of results, including results similar to some of those in \cite{GalGraf} (but in this higher dimensional setting),  in which the analysis is substantially simplified and more causal theoretic in nature.  By our conformal boundary-based definition, a spacetime can have arbitrarily many ${\rm AdS}_2 \times S^{n-1}$ ends.  But when the null energy condition is assumed to hold, it is shown that spacetime $(M,g)$ can have at most two ${\rm AdS}_2 \times S^{n-1}$ ends.  Further, for spacetimes  $(M^{n+1},g)$ which satisfy the null energy condition and which have two communicating  ${\rm AdS}_2 \times S^{n-1}$ ends, it is then shown that 
(i) $(M,g)$ has spatial topology that of an $n$-sphere minus two points, (ii) $(M,g)$ admits two transverse foliations by totally geodesic null hypersurfaces, and (iii) the intersections of these foliations give rise to a foliation of spacetime by totally geodesic $n-1$ spheres (not necessarily round).  Hence, under the present asymptotic assumptions,  (ii) and (iii) extend results in \cite{GalGraf} to this higher dimensional situation. A result on topological censorship for such spacetimes is also obtained. 

While the idea behind Definition \ref{maindef} comes from exact $\text{AdS}_2 \times S^2$ and our main examples are $\text{AdS}_2\times S^2$ and related spacetimes, this definition can also be read more generally as merely describing the notion of a  one-dimensional conformal  timelike infinity and  
it is perhaps possible that this new version of a `singular' timelike conformal boundary can be extended in some manner to other situations.

\bigskip
The paper is organized as follows: In Section \ref{Def sec} we define what we mean by a smooth spacetime having $k$ asymptotically $\text{AdS}_2\times S^{n-1}$ ends. Since this definition involves non-smooth metrics, we include an appendix on low-regularity causality theory (Appendix \ref{app:lip}). We then proceed to use the definition to derive basic properties of null geodesics and give a brief overview of certain classes of examples (which are discussed in more detail in Appendix \ref{app:examples}).

In Section \ref{sec:two ends}, we look at spacetimes satisfying the null energy condition with exactly two asymptotically $\text{AdS}_2\times S^{n-1}$ ends and derive the existence of totally geodesic null hypersurfaces and of a foliation by $(n-2)$-spheres. These results are similar to the ones in \cite{GalGraf}. Next, in Section \ref{sec:more ends}, we show that if the null energy condition holds $(M,g)$ can have at most two asymptotically $\text{AdS}_2\times S^{n-1}$ ends. 

Finally, we prove a version of topological censorship for spacetimes with $k$ asymptotically $\text{AdS}_2\times S^{n-1}$ ends in Section \ref{topcen}.

\section{Definition and basic properties}\label{Def sec}

${\rm AdS}_2 \times S^{n-1}$ can be expressed globally as the manifold $M = \R \times (0,\pi) \times S^{n-1}$,
with metric 
\begin{align}
 g &= \frac{1}{\sin^2x}(-dt^2 + dx^2) +  d \omega^2_{n-1}   \nonumber \\
 & = \frac{1}{\sin^2x}\left(-dt^2 + dx^2 + \sin^2x\,  d \omega^2_{n-1}\right)
\end{align}
where $d \omega^2_{n-1}$ is the round unit sphere metric on $S^{n-1}$.  Note that the metric  within the parentheses above is the metric $\ov{g}$ of the Einstein static universe defined on $\ov{M} = \bbR \times S^{n}$.

Hence, in this simple manner, we see that $(M,g)$ conformally embeds into the Einstein static universe  $(\ov{M}, \ov{g})$, with $g = \frac1{\Omega^2} \ov{g}$ and $\Omega = \sin x$. Further, $(\ov{M}, \ov{g})$ is a spacetime extension of $(M,\Omega^2g)$, with the latter missing two timelike lines $\R \times \{n\}$ and $\R \times \{s\}$, where $n,s$  are antipodal points on $S^{n}$.  By defining $\Omega = 0$ on these timelike lines,  $\Omega$ extends to a Lipschitz function on $\ov{M}$.  Although $d\Omega$ is not defined on these timelike lines, it remains bounded on approach to them.  One is led to think of each of these timelike lines as representing a lower dimensional timelike conformal infinity.  This situation motivates the following definition.\footnote{In this paper all manifolds are smooth.}

\medskip

\begin{Def} \label{maindef}
{\rm
 A smooth spacetime  $(M^{n+1},g)$ is said to have $k$ \emph{asymptotically $\text{AdS}_2\times S^{n-1}$ ends} (where $k \in \{1, 2, \dotsc, \infty\}$) provided there exists a spacetime $(\ov{M}, \ov{g})$, where $\ov{g}$ is a $C^{0,1}$ metric, and a function 
$\Omega \in C^{0,1}(\ov{M})$ such that the following hold:}
\ben
\item[{\rm(i)}] {\rm$(\ov{M}, \ov{g})$ is  globally hyperbolic.} 
\item[{\rm(ii)}] {\rm$M \subset \ov{M}$ and $\d M = \J$, where $\mathcal{J}$ is the disjoint union $\J = \bigsqcup_{i = 1}^k \J_i$ where each $\J_i$ is a smooth inextendible timelike curve in $\ov{M}$.}
\item [{\rm(iii)}]\label{omegabbd} 
{\rm$\Omega |_M \in C^{\infty}(M)$, $\Omega > 0$ on $M$ and $\Omega = 0$ on $\J$.  Further, suppose that for any point $p\in \J$ there exists a neighborhood $U$ of $p$ such that $U\cap \J$ is connected and $d\Omega$ remains bounded on 
$U\setminus \J$.\footnote{Technically boundedness of $d\Omega$ is already implied by assuming that $\Omega$ is Lipschitz on $\ov{M}$.}}
\item[{\rm(iv)}]  {\rm On $M$, $\ov{g} = \Omega^2 g$.  (In particular, $\ov{g}$ is smooth on $M$.)}
\een
\end{Def}

For example, in the sense of this definition, $\text{AdS}_2\times S^{n-1}$ has two asymptotically $\text{AdS}_2\times S^{n-1}$ ends. 

\smallskip

{\it Remark.} Note that $\J$ has to be closed: Any finite union of inextendible timelike curves in a strongly causal spacetime must be closed. This is no longer true for countably infinite unions; nevertheless, even in the countably infinite case, closedness of $\J $ is ensured by point (ii) in the definition. Also, (iii) implies that any $\J_i \subset \J=\bigcup_{l\in I} \J_l$ has a neighborhood $U$ in $\ov{M}$ that doesn't intersect any other $\J_l$, i.e., $U\cap \J_l=\emptyset $ for all $l\neq i$.

\smallskip

We also point out that in the standard treatment of conformal infinity, where $\Omega$ is smooth, with $\Omega =0$ and $d \Omega \ne 0$ on   $\J$, condition (iii)  is trivially satisfied.  We will use condition (iii) in conjunction with the following result proved in \cite{Graf}.
\begin{lem}\label{bddderivative}
	Let $(M,g)$ be a spacetime with Lipschitz metric. Assume $g$ is $C^1$ on some open subset $O\subset M$ and $\gamma :[0,1)\to O$ is a causal curve that is a solution of the geodesic equation. If $\gamma $ is continuously extendible to $p:=\gamma(1)\in \partial O$ then $\dot{\gamma}$ remains bounded on $[0,1)$.
\end{lem}

The following proposition extends to the present setting a basic result for spacetimes with conformal infinity in the conventional sense.

\begin{prop}\label{complete prop}
Suppose $(M,g)$ has an asymptotically $\emph{\text{AdS}}_2 \times S^2$ end $\J$. Let $\g:(0,a) \to M$ be a future directed $\ol{g}$-null geodesic with past end point $p \in \J$. Then $\g$ is past null complete as a $g$-geodesic.

\end{prop}

\proof 
Let $\g \colon (0,a) \to M$ be our $\bar g$-null geodesic with $\ov{g}$-affine parameter $\ov{s}$. By assumption, $\g$ extends continuously to $\g(0) = p$.  It is a standard fact  \cite{Wald} that  $\g|_M$ is a null $g$-geodesic with $g$-affine parameter $s$ satisfying $ds/d\ov{s} = c\Omega^{-2}$ for some constant $c$. Then, fixing $b \in (0,a)$,  $\g$ is past complete with respect to $g$ provided the integral
\begin{equation}\label{past complete integral}
\int_0^b \frac{1}{f(\bar s)}d\ov{s}
\end{equation}
diverges,  where $f(\bar s) = \Omega^2(\g(\bar s))$.  We have $f(\bar s) > 0$ for  $\bar s \in (0,a)$ and $f(0) = 0$.  Moreover, Lemma \ref{bddderivative} and condition (iii) imply that
there exists $A > 0$ such that $f'(\bar s)=2\Omega(\g(\bar s))d\Omega(\dot{\g}(\bar s)) \le A$ for all $\bar s \in  (0,a)$.  This implies that $f(\bar s) \le A \bar s$ for all $\bar s \in  (0,a)$, from which it follows that \eqref{past complete integral} diverges.\qed

\medskip
\noindent{{\bf Examples.}
We conclude this section by describing several examples of spacetimes with asymtotically ${\rm AdS}_2 \times S^2$ ends in the sense of of Definition \ref{maindef}.  

As a first example, let $(\ol{M}, \ol{g})$ be $n+1$ dimensional Minkowski space with standard coordinates $(x^0 = t, x^1,\cdots,x^n)$.  Let $M = \ol{M} \setminus \{t- \text{axis}\}$, with metric 
$$
g = \frac{\ol{g}}{\Omega^2}
$$ 
where $\Omega \in C^{\infty}(M)$, $\Omega > 0$, and near the $t$-axis, $\Omega = |\vec{x}| = \sqrt{\sum_{i=1}^n (x^i)^2}$.  Then $(M,g)$ has exactly one asymptotically ${\rm AdS}_2 \times S^{n-1}$ end, and $\cal{J} = $ the $t$-axis.  By removing other $t$-lines, this example shows that there are spacetimes with countably many asymptotically 
${\rm AdS}_2 \times S^{n-1}$ ends.   However, as shown in Section 4, if one assumes that the null energy condition (NEC) holds, there can be at most two asymptotically ${\rm AdS}_2 \times S^{n-1}$ ends.  In the appendix of \cite{Tod1}, Paul Tod presents an interesting example of a four dimensional spacetime with exactly one asymptotically ${\rm AdS}_2 \times S^{2}$ end, which, unlike the example above, satisfies then NEC. As described in \cite{Tod1},  this example is a solution of the Einstein equations, with source term the sum of a charged dust and an electomagnetic field.  

In \cite{Tod2} Tod presents a class of examples with metric of the form 
\beq
g = \frac{e^{-2f(t,x)}}{\sin^2x} (-dt^2 + dx^2)  + d\omega_2^2  \,,
\eeq
which satisfy the NEC and the asymptotic conditions assumed in \cite{GalGraf}.  For suitable choices of $f$, these are  examples of spacetimes satisfying the NEC with two asymptotically ${\rm AdS}_2 \times S^2$ ends, as defined here.

Next, we consider examples of the following type:  $M = \bbR^2 \times S^{n-1}$, with metric,
\beq\label{rcoords}
g = \underbrace{-f(r)dt^2 + \frac1{f(r)} dr^2}_{g_1}  +  \, \underbrace{d \omega^2_{{n-1}_{{}_{{}_{}}}}}_{g_2}
\eeq
where $f \in C^{\infty}(\bbR)$,  $f > 0$. In appendix \ref{app:examples} we show ${\rm AdS}_2 \times S^{n-1}$ corresponds to the choice $f(r) = r^2 +1$.  
We also find conditions on $f(r)$ which ensure that $(M,g)$ has two asymptotically ${\rm AdS}_2 \times S^{n-1}$ ends in the sense of Definition \ref{maindef}.

The Ricci tensor of $(M,g)$ is given by,
\beq\label{eq:ric}
\ric = (Kg_1) \oplus g_2  \,,
\eeq
where $K$ is the Gaussian (i.e. sectional) curvature of the $t$-$r$ plane.  A computation shows,\beq\label{eq:sec}
K = -\frac12 \frac{\d^2 f}{\d r^2}  \,.
\eeq

We wish to consider circumstances under which the NEC holds. Let $\{e_0 = \d_t/|\d_t|, e_1 =  \d_r/|\d_r|, e_2, \cdots, e_n\}$  be an orthonormal basis for $T_pM$.  Then, it follows from \eqref{eq:ric} and \eqref{eq:sec}, that, for any null vector $X = \sum_{i = 0}^n X^i e_i \in T_pM$, 
\begin{align}
\ric(X,X) &=(1 - K)\, (\sum_{i = 2}^n (X^i)^2) \nonumber \\
&= \left(1 + \frac12 \frac{\d^2 f}{\d r^2}\right) \, \left(\sum_{i = 2}^n (X^i)^2 \right)  \,.
\end{align}
Thus, the NEC holds at all points where $\frac{\d^2 f}{\d r^2} \ge -2$.

\smallskip
Now choose $f(r)$ as follows:
\ben
\item $f(r)$ even, $f(r) = f(-r)$, for all $r$.
\item $f$ is (weakly) concave up, i.e., $f''(r) \ge 0$ for all $r$.
\item $f(r) = 1+  r^2$ outside some interval $[-r_0, r_0]$.
\een
For such a choice (of which there are many), $(M,g)$ satisfies the NEC and is {\it exactly} $\text{AdS}_2\times S^{n-1}$ outside $[-r_0, r_0]$. 

We mention, as a last example, Schwarzschild-${\rm AdS}_2 \times S^{n-1}$, by which we mean, $M' = M \cap\{r > 0\}$, with metric \eqref{rcoords} where  $f(r)$ is given by,
\beq\label{SSAdS_2}
f(r) = 1 - \frac{2m}{r} + r^2 \,  ,
\eeq
see appendix \ref{app:examples}. Here we allow $f(r)$ to be negative as well as positive.  $(M',g)$ is an asymptotically ${\rm AdS}_2 \times S^{n-1}$ black hole spacetime, with horizon located at the single positive root $r_*$ of $f(r)$ (at which there is a coordinate singularity; see the recent review \cite{Socolovsky} for a nice treatment of this).   One checks that the NEC holds on the region $r \ge m^{\frac13}$, which includes the domain of outer communications $r > r_*$.

\section{Asymptotically $\text{AdS}_2{\bf \times S^{n-1}}$ spacetimes with two ends}\label{sec:two ends}

\subsection{Totally geodesic null hypersurfaces} \label{sec:onenullhypersurface}

\begin{prop}\label{null hypersurface prop}
Suppose $(M,g)$ has two asymptotically $\emph{\text{AdS}}_2 \times S^{n-1}$ ends $\J_1$ and $\J_2$. Let $p \in \J_1$. Suppose $J^+(p,\ov{M}) \cap \J_2 \neq \emptyset$.
Then

\begin{itemize}
\item[\emph{(1)}] there is a null curve $\eta \colon [0,1] \to \ov{M}$, contained in $\pd J^+(p, \ov{M})$, such that $\eta(0) = p$, $q :=\, \eta(1) \in \J_2$, and $\eta|_{(0,1)}$ is a complete null line in $M$,

\item[\emph{(2)}] $\pd J^+(p, \ov{M}) \,=\, \pd J^+(\eta, \ov{M})$,

\item[\emph{(3)}] if $\g \subset \pd J^+(\eta|_{(0,1)}, M)$ is past inextendible within $M$, then $\g$ has past endpoint $p$ or past endpoint $q$ 
within $\ov{M}$.

\end{itemize}
\end{prop}

\proof
We first prove (1). Note first that $J^+(p, \ov{M})$ cannot contain all of $\J_2$ (as otherwise a past inextendible portion of $\J_2$ would be imprisoned in a compact set {which contradicts strong causality \cite[Proposition 3.3]{LingCausal}). Together with $J^+(p, \ov{M})\cap \J_2 \neq \emptyset$ this implies that there exists a point $q \in  \pd J^+(p,\ov{M}) \cap \J_2 $. By Proposition \ref{bd J for GH prop} we have
\begin{equation}\label{bd J equation}
\pd J^+(p,\ov{M}) \,=\, J^+(p, \ov{M}) \setminus I^+(p, \ov{M})
\end{equation}
Therefore there is a causal curve $\eta\colon [0,1] \to \pd J^+(p,\ov{M}) $ with $\eta(0) = p$ and $\eta(1) = q$. Let $\hat{\eta} = \eta |_{(0,1)}$. Since $\J_1$ and $\J_2$ are disjoint, we can assume $\hat{\eta} \subset M$. 
We claim that $\hat{\eta}$ is achronal in $M$. Suppose not. Then there are two values $0 < t_1, t_2 < 1$ and a timelike curve from $\eta(t_1)$ to $\eta(t_2)$. Then the push-up property (Proposition \ref{Push-up Prop}) implies $\eta(t_2) \in I^+(p, \ov{M})$. This contradicts $\eta(t_2) \in \pd J^+(p, \ov{M})$ which proves the claim. Therefore $\hat{\eta}$ is an achronal inextendible null geodesic in $M$. Hence it's a null line. Completeness follows from Proposition \ref{complete prop}.

Now we prove (2). Recall that $\pd J^+ = J^+ \setminus I^+$. Therefore it suffices to show $J^+(p,\ov{M}) = J^+(\eta, \ov{M})$ and likewise with $I^+$. The equality for $J^+$ holds trivially since $\eta(0) = p$. To show $I^+(p, \ov{M}) = I^+(\eta, \ov{M})$, note that the inclusion $\subseteq$ follows trivially and the inclusion $\supseteq$ follows by the push-up property (Proposition \ref{Push-up Prop}).

Finally we prove (3).  Suppose $\g \colon (0,1) \to \pd J^+(\hat{\eta}, M)$ is past inextendible within $M$. By closedness of $J^+(p,\ov{M})$ and (2), 
$\gamma \subseteq J^+(p,\ov{M})$, so for $t\in (0,t_0)$, $\gamma(t) \subseteq J^+(p,\ov{M})\cap J^-(\gamma(t_0),\ov{M})$. Hence $\gamma $ is past imprisoned in a compact subset of $\ov{M}$, so $\gamma $ extends continuously to $p':=\gamma(0)\in \ov{M}$. By assumption, $\gamma $ is past inextendible in $M$, so $p'\not \in M$, i.e., $p'\in \J_1\cup \J_2$. Now we show that this implies $p' = p$ or $p' = q$. First suppose $p' \in \J_2$. If $p' \in I^+(q, \ov{M})$, then for sufficiently small $t$, we have $\gamma(t) \in I^+(q, \ov{M})$. Therefore there is a timelike curve from $p$ to $\gamma(t)$ by the push-up property. This contradicts equation (\ref{bd J equation}). Now suppose $p' \in I^-(q, \ov{M})$. Then for small $t$ we have $\gamma(t) \in I^-(q, \ov{M})$, but this gives a timelike curve from $p$ to $q$ -- another contradiction. Therefore $p' \in \J_2$ implies $p' = q$. Now suppose $p' \in \J_1$. If $p' \in I^+(p, \ov{M})$, then for small $t$ we can find a timelike curve from $p$ to $\gamma(t)$ -- a contradiction. If $p' \in I^-(p, \ov{M})$, then for small $t$ we can find a closed causal curve from $\gamma(t)$ to $p$ to $\gamma(t)$ which contradicts causality of $(\ov{M}, \ov{g})$.  Therefore $p' \in \J_1$ implies $p' = p$. \qed

\medskip
\medskip

\begin{thm}\label{thm:null hypersurface}
	Assume the hypotheses of Proposition \ref{null hypersurface prop}. If $(M,g)$ satisfies the null energy condition, then there exists a totally geodesic null hypersurface in $M$ containing $\eta|_{(0,1)}$. 
\end{thm}

\proof
This is a consequence of Theorem 4.1 in \cite{GalMax}. Note that by Remark 4.2 in \cite{GalMax}, it suffices that the generators of $\pd J^+\big(\eta|_{(0,1)}, M\big)$ are past complete and the generators of $\pd J^-\big(\eta|_{(0,1)}, M\big)$ are future complete.
This follows from Proposition \ref{complete prop} and (3) in Proposition \ref{null hypersurface prop} (along with its time-dual statement).
\qed

\medskip
\medskip

\begin{lem}\label{boundary lem}
Under the hypotheses of Proposition \ref{null hypersurface prop}, let $\eta \colon [0,1] \to \ov{M}$ be the constructed null curve. Then 
\[
\pd J^+(\eta|_{(0,1)}, M) \,=\, \pd J^+(\eta, \ov{M}) \cap M.
\]
\end{lem}

\proof
Denote $\hat{\eta} = \eta|_{(0,1)}$. Recall that $\hat{\eta} \subset M$. Since $(M,g)$ is smooth, Corollary \ref{bd for Lipschitz cor} gives 
$
\pd J^+(\hat{\eta}, M) \,=\, \ov{J^+(\hat{\eta}, M)}^M \,\setminus\, I^+(\hat{\eta}, M).
$
By Proposition  \ref{bd J for GH prop} we have 
\begin{align*}
\pd J^+(\eta, \ov{M}) \cap M \,&=\, \big[J^+(\eta, \ov{M}) \setminus I^+(\eta, \ov{M}) \big] \cap M
\\
&=\, \big[J^+(\eta, \ov{M}) \cap M \big] \,\setminus\, \big[I^+(\eta, \ov{M}) \cap M \big].  
\end{align*}
Thus it suffices to show 
\begin{enumerate}
\item[(1)] $I^+(\hat{\eta}, M) \,=\, I^+(\eta, \ov{M}) \cap M$
\item[(2)] $\ov{J^+(\hat{\eta}, M)}^M \,=\, J^+(\eta, \ov{M}) \cap M $

\end{enumerate}

We first prove (1). The left inclusion $\subseteq$ is clear. To show $\supseteq$, first note the following: For any $x\in \J_i$ ($i\in\{1,2\}$ fixed) there exists a neighborhood $U$ of $x$ in $\ov{M}$ such that for any $y_1,y_2\in \J_i \cap U$ with $y_2 \in I^+(y_1, \ov{M})$, $y_1$ and $y_2$ have neighborhoods $V_1,V_2\subseteq U$ such that any $y'_2\in V_2$ can be reached from any $y'_1\in V_1$ by a future directed timelike curve $\gamma$ in $U$ such that $\gamma \cap \J_i=\{y'_1, y'_2\}\cap \J_i$ (possibly empty). (This can be seen, e.g., by choosing a cylindrical neighborhood as in \cite{CG} such that additionally $U\cap \J_l =\emptyset $ for $l\neq i$ and such that $\J_i$ is the $x^0$-coordinate line.) From this we see that any two points $p_1,p_2$ on $\J_i$ have neighborhoods $U_1,U_2$ such that any point in $U_1 \cap M$ can be connected to any point in $U_2 \cap M$ by a timelike curve lying entirely in $M$. 

 Continuing the proof of (1), take $z \in I^+(\eta, \ov{M}) \cap M$. Since $I^-(z, \ov{M})$ is open, there is a timelike curve $\g \colon [0,1] \to \ov{M}$ such that $\g(0) \in \hat{\eta}$ and $\g(1) = z$. If $\gamma$ is contained in $M$, then we're done. If $\gamma$ intersects just one $\J_i$, say $\J_1$, then let  $s_0:=\min \{s:\gamma(s)\in \J_1\}$ and $s_1:=\max \{s:\gamma(s)\in \J_1\}$. By the above paragraph,  for small enough $\epsilon>0$ the curve $\gamma|_{[s_0-\epsilon,s_1+\epsilon]}$ can be replaced by a timelike curve with the same endpoints contained entirely in $M$. Hence $z \in I^+(\hat{\eta}, \ov{M})$. If $\gamma$ intersects multiple $\J_i's$, then we can repeat the procedure above to get $z \in I^+(\hat{\eta}, \ov{M})$. This proves (1).
 
Now we prove (2). We have 
\[
\ov{I^+(\hat{\eta},M)}^M \,=\, \ov{I^+(\eta, \ov{M})\cap M } ^M \,=\, \ov{I^+(\eta, \ov{M})\cap M}^{\ov{M}}\cap M \,=\, \ov{I^+(\eta, \ov{M})}^{\ov{M}}\cap M .
\]
The first equality follows from (1). The second equality follows from basic point set topology. For the third equality, $\subseteq$ is trivial and $\supseteq$ follows because $M$ is dense in $\ov{M}$. Thus (2) follows from Corollary \ref{bd for Lipschitz cor} and the fact that $J^+(\eta, \ov{M})$ is closed since $(\ov{M}, \ov{g})$ is globally hyperbolic \cite[Proposition 3.5]{LingCausal}.
\qed

\medskip
\medskip

\begin{prop} \label{cor: bd J+ is bd J-} 
Under the hypotheses of Theorem \ref{thm:null hypersurface}, let $\eta \colon [0,1] \to \ov{M}$ be the constructed null curve from Proposition \ref{null hypersurface prop}. Let $p = \eta(0)$ and $q = \eta(1)$. Then 
\begin{align*}
\pd J^+(p, \ov{M}) \,=\, \pd J^+(\eta|_{(0,1)}, M) \cup \{p, q\} \,=\, \pd J^-(\eta|_{(0,1)}, M) \cup \{p, q\} \,=\, \pd J^-(q, \ov{M})
\end{align*}
\end{prop}

\proof
Denote $\hat{\eta} = \eta|_{(0,1)}$. First note that by achronality both $\J_1$ and $\J_2$ can intersect $\pd J^+(\eta,\ov{M})$ only once, hence $\pd J^+(\eta,\ov{M})=(\pd J^+(\eta,\ov{M})\cap M)\cup \{p,q\}$. So $\pd J^+(p,\ov{M})=\pd J^+(\hat{\eta},M)\cup \{p,q\}$ follows from Proposition \ref{null hypersurface prop} (2) and Lemma \ref{boundary lem}.

Secondly, $\pd J^+(\hat{\eta},M)$ has only one connected component: From the above we know that $\pd J^+(\hat{\eta},M)=\pd J^+(p,\ov{M})\setminus \{p,q\}.$ As the boundary of a future set $\pd J^+(p,\ov{M})$ is a $C^0$ hypersurface \cite[Corollary 4.8]{LingCausal} so $\pd J^+(p,\ov{M})\setminus \{p,q\}$ is connected because $\pd J^+(p,\ov{M})$ is (any point on $\pd J^+(p,\ov{M})\subset J^+(p,\ov{M})$ can be connected to $p$). Thus, Remark 4.2 in \cite{GalMax} gives $\pd J^+(\hat{\eta},M)=\pd J^-(\hat{\eta},M)$, establishing the remaining equalities (using time duality for the last).
\qed

\begin{thm} \label{thm: pseudo cauchy}\label{thm:topology} Under the assumptions of Theorem \ref{thm:null hypersurface} let $\eta $ be an achronal null curve as constructed in Proposition \ref{null hypersurface prop}. Then 
	\begin{enumerate}
		\item[\emph{(1)}] any inextendible causal curve in $\ov{M}$ must meet $\pd J^+(\eta(0),\ov{M})$ and any inextendible timelike curve in $\ov{M}$ intersects $\pd J^+(\eta(0),\ov{M})$ exactly once,
		\item[\emph{(2)}] $\partial J^+(p,\ov M)$ is homeomorphic to $S^{n}$,
		\item[\emph{(3)}] the Cauchy surfaces of $\ov{M}$ are homeomorphic to $S^n$ and 
		\item[\emph{(4)}] $\ov{M}$ is  homeomorphic to $\R \times S^{n}$.
	\end{enumerate}	
\end{thm}
\proof
We start by proving (1): Let $p:=\eta(0)$ and and let $\check{g}_\epsilon $ be a smooth metric on $\ov{M}$ with narrower lightcones than $g$, i.e., $\check{g}_\epsilon\prec g$ (c.f., \cite{CG}). Then $(\ov{M},\check{g}_\epsilon)$ is globally hyperbolic and $S:=\pd J^+(p,\ov{M})$ is a $\check{g}_\epsilon$-acausal
compact (since it is contained in the compact diamond $J^+(p,\ov{M})\cap J^-(\eta(1),\ov{M})$) topological hypersurface, hence a $\check{g}_\epsilon$-Cauchy hypersurface (cf., Theorem \ref{acausal thm} in Appendix A). This implies that $\ov{M}=I^+_{\check{g}_\epsilon}(S)\cup S \cup I^-_{\check{g}_\epsilon}(S)$. 
Since $I^{\pm}_{\check{g}_\epsilon}(S)\subseteq I^{\pm}(S)$ this implies that $\ov{M}=I^+(S)\cup S \cup I^-(S)$. Hence any inextendible causal curve $\g :(0,1)\to \ov{M}$ not meeting $S$ must intersect at least one of $I^+(S)$ or $I^-(S)$. Assume that $\g$ meets $I^+(S)$ in $\g(s_0)$ but doesn't intersect $S=J^+(\eta(0))\setminus I^+(\eta(0))=J^+(S)\setminus I^+(S)=\pd J^+(S)$, then $\g|_{(0,s_0]}$ is imprisoned in the compact set $J^-(\g(s_0))\cap J^+(S)$, a contradiction. 

Next, we establish (2): Let $U=I\times V$ be a cylindrical neighborhood around $p\equiv (0,\bar{p})$ (as defined in \cite{CG}), then $\pd J^+(p,U)$ is a Lipschitz graph over $V$ (cf.~\cite[Prop.~1.10]{CG}). Together with achronality of $\pd J^+(p,\ov{M})$ this implies $\pd J^+(p,\ov{M})\cap U=\pd J^+(p,U)$ (since any curve $t\mapsto (t,x_0)\in I\times V$ meets $\pd J^+(p,\ov{M})$ at most once and $\pd J^+(p,U)$ exactly once), so 
$\pd J^+(p,\ov{M})\cap U$ is a Lipschitz graph over $V$, hence homeomorphic to $V$ which is diffeomorphic to 
$\R^{n}$. Similarly, there exists a neighborhood $U_q$ around $q:=\eta(1)$ such that $\pd J^-(p,\ov{M})\cap U_q=\pd J^-(q,\ov{M})\cap U_q$ is homeomorphic to $\R^{n}$.
Now, let $\mathcal{S}\subseteq V$ be a smooth sphere around $\bar{p}$, then $(I\times \mathcal{S})\cap \pd J^+(p,\ov{M})\subseteq M$ is a smooth $n-1$-dimensional submanifold homeomorphic to $S^{n-1}$ that is met exactly once by each null geodesic generator of $\pd J^+(\eta|_{(0,1)},M)$. As in the discussion in \cite{GalMax}, Remark~4.1, this implies that $\pd J^+(p,\ov{M})\setminus \{p,q\} =\pd J^+(\eta|_{(0,1)},M)$ is homeomorphic to $\R \times S^{n-1}$. Together with the description of $\pd J^+(p,\ov{M})$ near $p$ and $q$ this shows that $\pd J^+(p,\ov{M})$ is indeed homeomorphic to $S^{n}$.

Now (3) and (4) follow by noting that both $\pd J^+(\eta(0),\ov{M})$ and any Cauchy surface $\Sigma$ of $\ov{M}$ are Cauchy surfaces for  $\check{g}_\epsilon $, hence they are homeomorphic, so  $\Sigma\cong S^n$ and $\ov{M}\cong \R \times S^n$.
\qed

\subsection{Foliations}

In section \ref{sec:onenullhypersurface} we established the existence of a totally geodesic achronal  null hypersurface for spacetimes obeying the NEC with two asymptotically $\text{AdS}_2 \times S^{n-1}$ ends $\J_1$ and $\J_2$ if $J^+(\J_1)\cap \J_2 \neq \emptyset$, i.e., if the ends are able to communicate. In this section we will show the existence of two transversal foliations by totally geodesic achronal null hypersurfaces under the following stronger assumption on the causal relationship between both ends. 

\medskip

\begin{Def}\emph{
	Two asymptotically $\text{AdS}_2 \times S^{n-1}$ ends $\J_1$ and $\J_2$ are said to be \emph{communicating at all times} if $\J_1\subseteq J^\pm(\J_2)$ and $ \J_2 \subseteq J^\pm (\J_1)$.}
\end{Def}

\medskip

\begin{prop}\label{prop:nullfoliation}
	Suppose $(M,g)$ has two asymptotically $\emph{\text{AdS}}_2 \times S^{n-1}$ ends $\J_1$ and $\J_2$ that are communicating at all times. If $(M,g)$ satisfies the NEC, then it is continuously foliated by totally geodesic achronal null hypersurfaces $N_t=\partial J^+(\gamma(t),\ov M)\cap M$, where $\gamma :\R \to \J_1$ is a paramerization of $\J_1$.
\end{prop}

\begin{proof}
	Since $\J_1$ and $\J_2$ are communicating at all times, $J^+(p, \ov M)\cap \J_2 \neq \emptyset $ for all $p\in \J_1$. Thus we may apply Proposition \ref{null hypersurface prop} and Theorem \ref{thm:null hypersurface} to  each point on $\J_1$ to obtain a family $\{N_t\}_{t\in \R}$ of totally geodesic achronal null hypersurfaces satisfying $N_t=\partial J^+(\gamma(t),\ov M)\cap M$, where $\gamma :\R \to \J_1$ is a paramerization of $\J_1$. We show that this is a continuous foliation.
	
	Since $\gamma(t_1)\ll \gamma(t_2)$ if $t_1< t_2$ we have $N_{t_1}\cap N_{t_2}=\emptyset $ if $t_1\neq t_2$. Next we show that $\{N_t\}_{t\in \R}$  covers all of $M$. Let $x\in M$. Suppose $x\in J^+(N_t, \ov M)$ for some $t$.  Then $J^-(x, \ov M)\cap \J_1 \neq \emptyset $ and, since $J^-(x, \ov M)$ cannot contain all of $\J_1$ (as otherwise a future inextendible portion of $\J_1$ would be imprisoned in a compact set)	there exists $t_x$ such that $\gamma(t_x)\in \partial J^-(x, \ov M)\cap \J_1 \neq \emptyset $. This implies that $x\in N_{t_x}$. Suppose, on the other hand, $x\notin  J^+(N_t, \ov M)$ for any $t$.  Then, by Theorem \ref{thm: pseudo cauchy}, $x\in I^-(N_t, \ov M)=I^-(q_t, \ov M)$, where $q_t:=\J_2\cap N_t$, for all $t\in \R$. Let $\eta_t:[0,1]\to \ov M $ denote one of the achronal past directed null geodesic generators of $N_t$ starting at $q_t$ and ending at $\gamma(t)$. Then $\{\eta_t(0)\}_{t\leq 0}$ is contained in a compact subset of $\J_2$, hence there exists a sequence $t_n\to -\infty$ such that $\eta_{t_n}(0)$ converges to some $\ov q \in \J_2$. Further, since $t_n\to -\infty$, $\eta_{t_n}(1)$ leaves every compact subset of $\ov M$. Hence \cite[Theorem 3.1]{MinguzziLimitCurves}\footnote{While the cited result is for smooth metrics, the same remains true for merely continuous metrics. This essentially follows from applying the smooth result to metrics with wider lightcones and then using a separate argument to show that the obtained limit curve is causal, see the proof of \cite[Thm.~1.5]{Samann}.} implies that there exists an inextendible \emph{achronal} limit curve $\ov \eta $ starting at $\ov q$ leaving every compact subset of $\ov M$ which contradicts $\partial J^-(\ov q,\ov M) \cong S^{n-1}$ which follows from the time dual of Theorem \ref{thm:topology} (note that $J^-(\ov q,\ov M)\cap \J_1\neq \emptyset$ by assumption).
	
	So we have established that each $x \in M$ belongs to a unique totally geodesic null hypersurface $N_{t_x}$.	It remains to show that the map $x \mapsto t_x$ is continuous. Let $x_n \to x$ and let $\eta_n :[0,1]\to \ov M$ be the achronal null geodesic generator of $N_{t_{x_n}}$ from $\gamma(t_{x_n})$ to $x_n$. Let $x^+, x^-$ be points for which $x^-\leq x_n,x$ and $x_n,x \leq x^+$ for all $n$. Then $t_{x^-}\leq t_{x_n}\leq t_{x^+}$. Let $t_0$ be any accumulation point of the sequence $t_{x_n}$. By \cite[Theorem 3.1]{MinguzziLimitCurves} there exists a subsequence $\eta_{n_k}$ converging to an achronal null curve from $\gamma(t_0)=\lim \eta_{n_k}(0)$ to $x=\lim \eta_{n_k}(1)$. But this implies $x\in \partial J^+(\gamma(t_0), \ov M)\cap M$, hence $t_0=t_x$. So the sequence $t_{x_n}$ has $t_x$ as its only accumulation point, hence converges to $t_x$ and thus $x\mapsto t_x$ is continuous. 
\end{proof}

{\it Remark.} While we assumed that the ends are communicating at all times in Proposition \ref{prop:nullfoliation} this result actually only required that $\J_1 \subseteq J^-(\J_2)$ (to ensure that $J^+(p, \ov M)\cap \J_2 \neq \emptyset $ for all $p\in \J_1$) and $\J_2 \subseteq J^+(\J_1)$ (to ensure that $J^-(\ov q, \ov M)\cap \J_1 \neq \emptyset $ for all $\ov q\in \J_2$). However, in the next Theorem we need to use a second foliation by null hypersurfaces transverse to the first and to obtain this transverse foliation we will need that $\J_1 \subseteq J^+(\J_2)$ and $\J_2 \subseteq J^-(\J_1)$. Thus, Theorem \ref{thm:surfacefoliation} really needs the full definition of both ends communicating at all times. 

\medskip

\begin{thm}\label{thm:surfacefoliation}
	Suppose $(M,g)$ has two asymptotically $\emph{\text{AdS}}_2 \times S^{n-1}$ ends $\J_1$ and $\J_2$ that are communicating at all times. If $(M,g)$ satisfies the NEC, then it is continuously foliated by totally geodesic
	$(n-1)$-dimensional submanifolds homeomorphic to $S^{n-1}$. These submanifolds may be obtained as the intersections of two transverse foliations by totally geodesic achronal null hypersurfaces.
\end{thm}
\begin{proof}
	Let $N_t:=\partial J^+(\gamma_1(t),\ov M)\cap M$, where $\gamma_1 :\R \to \J_1$ is a paramerization of $\J_1$ and $\hat{N}_s:=\partial J^+(\gamma_2(s),\ov M)\cap M$, where $\gamma_2 :\R \to \J_2$ is a paramerization of $\J_2$ and let $S_{t,s}:=N_t\cap \hat{N}_s$. If non-empty, this intersection is a totally geodesic
	$(n-1)$-dimensional submanifold of $M$ because $N_t$ and $\hat{N}_s$ always intersect transversally. Additionally, every null geodesic generator of $N_t$ must meet $S_{t,s}$ exactly once: It suffices to show that  every null geodesic generator of $N_t$ must meet $\hat{N}_s$ exactly once if $N_t\cap \hat{N}_s\neq \emptyset$. For this intersection to be nonempty we must have $\partial J^+(\gamma_2(s),\ov M) \cap \J_1 \gg \gamma_1(t)$ and $\partial J^+(\gamma_1(t),\ov M) \cap \J_2 \gg \gamma_2(s)$, hence every null geodesic generator of $N_t$ starts in $I^-(\hat{N}_s,\ov M)$ and ends in $I^+(\hat{N}_s,\ov M)$, hence must intersect $\partial J^+(\hat{N}_s,\ov M)\cap M=\hat{N}_s$.
	
	So we may use the flow along the null geodesic generators of $N_t$ to conclude that $S_{t,s}$ is homeomorphic to any other $(n-1)$-dimensional submanifold of $N_t$ that is met exactly once by every null geodesic generator. Thus, remembering that by the last paragraph in the proof of (2) in Theorem \ref{thm:topology} the null hypersurface $N_t$ always contains an $(n-1)$-dimensional submanifold homeomorphic to $S^{n-1}$, every $S_{t,s}$ is homeomorphic to $S^{n-1}$.
	
	That this is indeed a continuous foliation follows completely analogously to the last paragraph in the proof of Theorem 3.16 in \cite{GalGraf}.
\end{proof}

\section{Asymptotically $\text{AdS}_2{\bf \times S^{n-1}}$ spacetimes with more than two ends}\label{sec:more ends}

\begin{thm}\label{thm:only2}
	Suppose $(M,g)$ has $k$ or countably infinite asymptotically $\emph{\text{AdS}}_2 \times S^{n-1}$ ends $\{\J_l\}_{l\in I}$, $I=\{1,\dots,k\}$, or $I=\mathbb{N}$. Suppose further that there exist $i,j\in I$, $i\neq j$, such that $J^+(p,\ov{M}) \cap \J_j \neq \emptyset$ for some $p\in \J_i$.  If $(M,g)$ obeys the null energy condition then $I$ is finite and $k = 2$; i.e. $(M,g)$ can have at most two asymptotically $\emph{\text{AdS}}_2 \times S^{n-1}$ ends.
\end{thm}

\proof W.l.o.g.~$i=1$, $j=2$.   Assume $k:=|I|>2$ (we include the case $k=\infty$).
Proceeding as in Proposition \ref{null hypersurface prop}, the assumptions allow us to obtain the first part of Proposition \ref{null hypersurface prop}:
\begin{itemize}
	\item[\emph{(1)}] There is an achronal null curve $\eta \colon [0,1] \to \ov{M}$ such that $\eta(0) = p\in \J_1$, $\eta(1) \in \J_2$.
\end{itemize}
		
	We may w.l.o.g.~assume that $\eta $ doesn't intersect any other $\J_l$: If $k$ is finite, $\eta $ intersects $\bigcup_{l=1}^k \J_l$ in isolated points, so this can be achieved by cutting $\eta$ off at the earliest parameter at which it intersects another $\J_l$,  renumbering the $\J_l$'s and rescaling $\eta$. If $k=\infty$, i.e., $I=\mathbb{N}$, note first that $\eta \cap M \neq \emptyset$ since $\eta\cap \bigcup_{l\in \mathbb{N}} \J_l $ contains at most countably infinitely many points. Let $t_0\in (0,1)$ be such that $\eta(t_0)\in M$, and set $t_1:=\inf \{t\in [0,t_0]: \eta|_{(t,t_0]}\subset M \}$ and $t_2:=\sup \{t\in [t_0,1]: \eta|_{[t_0,t)}\subset M \}$. Then, by closedness of $\J :=\bigcup_{l\in I} \J_l $ (cf. the remark after Definition \ref{maindef}) we have $\eta (t_1),\eta(t_2)\in \J $ and the desired property can again be achieved by cutting off $\eta $ (this time at $t_1$ and $t_2$), rescaling and renumbering the $\J_l$'s.
	
	Thus, $\eta|_{(0,1)}\subseteq M$ and $\eta|_{(0,1)}$ is a complete null line in $M$. With this we obtain the following versions of the other two parts of Proposition \ref{null hypersurface prop}:
	\begin{itemize}
	\item[\emph{(2)}] $\pd J^+(p, \ov{M}) \,=\, \pd J^+(\eta, \ov{M})$ and
	
	\item[\emph{(3)}] if $\g \subset \pd J^+(\eta|_{(0,1)}, M)$ is past inextendible within $M$, then $\g$ has past endpoint on $\J_l$ for some $l\in I$. In particular, all generators of $\pd J^+\big(\eta|_{(0,1)}, M\big)$ are past complete (and by time duality the generators of $\pd J^-\big(\eta|_{(0,1)}, M\big)$ are future complete).
\end{itemize}
Analogous to Theorem \ref{thm:null hypersurface}, Lemma \ref{boundary lem}, Proposition \ref{cor: bd J+ is bd J-} and Theorem \ref{thm: pseudo cauchy} we now may use this to get that $\pd J^+\big(\eta|_{(0,1)}, M\big)$ is a totally geodesic null hypersurface in $M$ and that $S:=\pd J^+(\eta(0),\ov{M})=\pd J^-(\eta(1),\ov{M})$ is a compact achronal hypersurface in $\ov{M}$ that is met exactly once by every inextendible timelike curve. Let $x\in S$ be the point where $\J_3$ intersects $S$. Then $x\neq \eta(0)$ and $x\neq \eta(1)$, $\eta(0)<_{\ov{M}}x$ and $x$ is the past endpoint of an achronal null curve $\tilde{\eta}:[0,1]\to \ov{M}$ ending in $q:=\eta(1)$. Again, w.l.o.g.~$\tilde{\eta}|_{(0,1)}\subseteq M$ is an achronal null line in $M$ (else, replace $x$) and hence $\tilde{S}:= \pd J^+(\tilde{\eta}(0),\ov{M})=\pd J^-(\tilde{\eta}(1),\ov{M})$ is also a compact achronal hypersurface in $\ov{M}$ that is met exactly once by every inextendible timelike curve. Let $p'$ be the point where $\J_1$ intersects $\tilde{S}$; hence,  $p'\in J^+(x,\ov{M})\setminus\{x\}$.   Further, since $\eta(1)=q=\tilde{\eta}(1)$, it must be that $\tilde{S}=S$.  It follows that $p' = \eta(0)$ since $\J_1$ intersects $S$ in $\eta(0)$.  This is a contradiction because we obtain $\eta(0)<_{\ov{M}} x <_{\ov{M}}p'=\eta(0)$.
\qed

\medskip

\section{Topological Censorship}\label{topcen}

Let $(M^{n+1},g)$ be a spacetime with $k \ge 1$  asymptotically ${\rm AdS}_2 \times S^{n-1}$ ends, and let $\J$ be one such end, i..e, $\J\equiv \J_1$ is one of the timelike lines $\J_1,\J_2, \dots $.  The domain of outer communications with respect to $\J$ is defined as the set
\beq\label{eq:doc}
\calD = I^-(\J, \ol{M}) \cap I^+(\J, \ol{M}).
\eeq
Note that $\calD$ is globally hyperbolic. This follows since for $p,q \in \mathcal{D}$, the push-up property can be used to show $J^+(p,\calD) \cap J^-(q, \calD) = J^+(p,\ol{M}) \cap J^-(q, \ol{M})$.

  Further, $\calD \cap M$ has $k\geq k_{\calD \cap M}\geq 1$ asymptotically ${\rm AdS}_2 \times S^{n-1}$ ends, one of which is $\J$. We are primarily interested in the case that $\calD$ is a proper subset of $\ol{M}$, which then signifies the presence of a black hole and/or white hole region.   

Roughly speaking, topological censorship asserts, under appropriate energy and causality conditions, that the topology (at the fundamental group level) of the domain of outer communications (DOC) can be no more complicated than the topology of conformal infinity.  In particular if conformal infinity is simply connected, then so is the DOC; see \cite{ChruGal} for a recent review.   Since, in our setting, $\J$ is simply connected, one would expect $\calD$ to be  simply connected.  This is almost the case.

\begin{thm}\label{topcenthm}
	Let $(M^{n+1},g)$ be a spacetime with $k \ge 1$ asymptotically ${\rm AdS}_2 \times S^{n-1}$ ends.  Let $\J$ be one such end, and  let $\calD$ be the DOC associated to $\J$.  Assume that the NEC holds on $\calD$.  Then either
	\ben
	\item[(i)] $\calD$ is simply connected, or
	\item[(ii)] $M$ has exactly one asymptotically ${\rm AdS}_2 \times S^{n-1}$ end (namely $\J$), $\calD \setminus \J$ contains a totally geodesic null hypersurface whose null geodesic generators have past and future end points on $\J$, and the Cauchy surfaces of $\calD$ (which are also Cauchy surfaces for $\ol{M}$) are double covered by an $n$-sphere. In particular, $\pi_1(\calD) = \bbZ_2$.
	\een
\end{thm}

Case (ii) can occur, as can be seen by considering a  simple quotient of the Einstein static universe.  By identifying antipodal points on the $n$-sphere, we obtain $\ol{M} = \bbR \times RP^{n-1}$, with obvious product metric.  Removing the $t$-line through the ``north pole" of $RP^{n-1}$, we obtain, after a conformal change, a spacetime $(M,g)$ which (i) has one asymptotically 
${\rm AdS}_2 \times S^{n-1}$ end, (ii) obeys the NEC, and (iii) has conformal completion $(\ol{M}, \ol{g})$ with Cauchy surfaces diffeomorphic to $RP^{n}$.

\proof[Proof of Theorem \ref{topcenthm}]   
Let $(\calD',g')$ be the universal covering spacetime of $(\calD, g)$, with covering map $p: \calD' \to \calD$, and $g' =p^*g$. Since $\J$ is simply connected, $\J' = p^{-1}(\J)$ consists of a disjoint union of copies of $\J$, $\J' = \sqcup_{\a \in A}\J_{\a}$, where $|A| = $ the number of sheets of the covering.   

Suppose first that
\beq\label{eq:empty}
I^-(\J_{\a}) \cap I^+(\J_{\b}) = \emptyset \text{ for all  }\a \ne \b \,.
\eeq
Consider the  collecton of open sets in $\calD'$,
\beq
U_{\a} = I^-(\J_{\a}) \cap I^+(\J_{\a})  \,.
\eeq
Equation \eqref{eq:empty} implies that the $U_{\a}$'s are pairwise disjoint.  It also implies that the $U_{\a}$'s cover 
$\calD'$:  Let $q'$ be any point  in $\calD'$, and consider $q = p(q') \in \calD$.  Equation~\eqref{eq:doc} implies that there exists a future directed timelike curve $\g$ from $\J$ to $\J$ passing through $q$. Lift $\g$ to obtain a timelike curve $\g'$ in $\calD'$, passing through $q'$, from $\J_{\a}$ to $\J_{\b}$ , for some $\a, \b$.  Equation \eqref{eq:empty} implies that $\a =  \b$, and hence $q' \in U_{\a}$.  Thus, since $\calD'$ is connected, there can be only one $U_{\a}$, and hence $\calD'$ is a one-sheeted covering of $\calD$, i.e. $\calD$ is simply connected.

Now suppose,
\beq\label{eq:nonempty}
I^-(\J_{\a}) \cap I^+(\J_{\b}) \neq \emptyset \text{ for some  }\a \ne \b \,.
\eeq
Let $\calK = \cup_{i =1}(\J_i)$, where $\J = \J_1, \J_2, \cdots $ are the asymptotically ${\rm AdS}_2 \times S^{n-1}$ ends of $\calD_0 :=\calD \cap M$, and let $\calK' = p^{-1}(\calK)$.   Since the fundamental group of any manifold is countable, $\calK'$ has countably many (perhaps countably infinite) components. 
Consideration of the function 
$\Omega' = \Omega \circ p$, where $\Omega : \calD \to \bbR$ is as in Definition \ref{maindef}, shows that each component of $\calK'$ is an asymptotically ${\rm AdS}_2 \times S^{n-1}$ end of $\calD_0' = p^{-1}(\calD_0)$.

From \eqref{eq:nonempty} it follows that  $J^+(\J_{\a})$ meets $\J_{\b}$.  
We further know that $\calD'$ is globally hyperbolic and satisfies the NEC, as these properties of $\calD$ lift to the cover.  
We may then apply Theorem \ref{thm:only2} to conclude that $\calD_0'$ has at most two 
${\rm AdS}_2 \times S^{n-1}$ ends.  In fact, it follows from \eqref{eq:nonempty} that $\calD'_0$ has exactly two 
${\rm AdS}_2 \times S^{n-1}$ ends, say, $\J_1'$ and $\J_2'$. If $p(\J_1')\neq p(\J_2')$, then $\calD_0$ has two asymptotically $AdS_2\times S^{n-1}$ ends and $\calD'$ is a one-sheeted covering of $\calD$. So $\calD$ is simply connected. 
If $\J=p(\J_1')\neq p(\J_2')$, then $\calD$ has only one end and $\calD'$ is a double cover of $\calD$. Moreover, by Theorem \ref{thm:null hypersurface}, there exists a totally geodesic null hypersurface $H'$ in 
$\calD' \setminus \J'$,  that extends from $\J_1'$ and $\J_2'$.  Then $H  = p(H')$ is an (immersed) totally geodesic null  hypersurface in $\calD$ whose null geodesic generators have past and future end points on $\J$.

Let $\S$ be a Cauchy surface for $\calD$. Then $\S' = p^{-1}(\S)$ is a Cauchy surface for $\calD'$. Using $\calD' \approx \bbR \times \S'$ and $\calD \approx \bbR \times \S$, it follows that $p|_{\S'} : \S' \to \S$ is a double covering of 
$\S$. 
Moreover, it follows from Theorem  \ref{thm:topology} that $\S'$ is homeomorphic to $S^{n-1}$.  Hence, $\S$ is double covered by a manifold homeomorphic to  $S^{n-1}$, and in particular is compact.  It follows that $\S$ must also be a Cauchy surface for $\ol{M}$ (cf., Theorem \ref{acausal thm} in Appendix~A).\qed

\appendix
\appendixpage
\addappheadtotoc

\section{Lipschitz metrics}\label{app:lip}

Recall that a \emph{$C^m$ spacetime} $(M,g)$ is a smooth manifold $M$ equipped with a $C^m$ time-oriented Lorentzian metric $g$. If a smooth spacetime $(M,g)$ has $k$ asymptotically $\text{AdS}_2\times S^{n-1}$ ends (see Definition \ref{maindef}), then the unphysical spacetime $(\ov{M}, \ov{g})$ has a $C^{0,1}$ metric $\ov{g}$, i.e. the components $\ov{g}_{\mu\nu} = \ov{g}(\pd_\mu, \pd_\nu)$ in any coordinate system $x^\mu$ are locally Lipschitz functions. In this case we say $(\ov{M}, \ov{g})$ is a \emph{Lipschitz spacetime}.

Classical references on causal theory such as \cite{Wald} and \cite{HE} make use of normal neighborhoods which require a $C^2$ metric. Therefore classical causal theory only holds for $C^2$ spacetimes. Since we are working with Lipschitz spacetimes, we require results from causal theory when the metric is only $C^{0,1}$. Treatments of causal theory for metrics with regularity less than $C^2$ can be found in \cite{CG, MinguzziCones, LingCausal}. 

Our definitions of timelike and causal curves will follow the conventions in \cite{LingCausal}. 
From that paper we have the following results. 

\medskip
\medskip

\begin{thm}[\cite{LingCausal}]
Let $(M,g)$ be a $C^0$ spacetime, then $I^+(p)$ and $I^-(p)$ are open.
\end{thm}

\medskip
\medskip

The following result is known as the \emph{push-up property} which is proved in \cite[Theorem 4.5]{LingCausal}. See also \cite[Lemma 1.15]{CG} 

\medskip
\medskip

\begin{prop}[\cite{CG, LingCausal}]\label{Push-up Prop}
Let $(M,g)$ be a Lipschitz spacetime. Then 
\[
I^+\big(J^+(p)\big) \,=\, I^+(p) \,.
\]
\end{prop}

\medskip
\medskip

The push-up property implies:

\medskip
\medskip

\begin{prop}
Let $(M,g)$ be a Lipschitz spacetime. Then 
\begin{itemize}
\item [$\emph{(1)}$] $\emph{\text{int}}\big[J^+(p)\big] \,=\, I^+(p)\,.$
\item [$\emph{(2)}$] $J^+(p) \,\subset\, \ov{I^+(p)} \,.$
\end{itemize}
\end{prop}

\proof
We first prove (1). Since $I^+(p) \subset J^+(p)$ and $I^+(p)$ is open, we have $I^+(p) \subset \text{int}\big[J^+(p)\big]$. Conversely, fix $q \in \text{int}\big[J^+(p)\big]$ and let $U \subset \text{int}\big[J^+(p)\big]$ be an open neighborhood of $q$. Let $q' \in I^-(q,U)$. Then there is a causal curve from $p$ to $q'$ and a timelike curve from $q'$ to $q$. Therefore $q \in I^+(p)$ by the push-up property. 

Now we prove (2). Fix $q \in J^+(p)$. Let $U$ be a neighborhood of $q$. Consider a point $q' \in I^+(q,U)$. Then $q' \in I^+(p)$ by the push-up property. 
\qed

\medskip
\medskip

\begin{cor}\label{bd for Lipschitz cor}
Let $(M,g)$ be a Lipschitz spacetime. Then 
\begin{itemize}
\item [$\emph{(1)}$] $\pd J^+(p) \,=\, \ov{J^+(p)} \,\setminus\, I^+(p)\,.$
\item [$\emph{(2)}$] $\ov{J^+(p)} \,=\, \ov{I^+(p)} \,.$
\end{itemize}
\end{cor}

\medskip
\medskip

Following \cite{LingCausal}, a $C^0$ spacetime $(M,g)$ is \emph{globally hyperbolic} provided it is strongly causal and $J^+(p) \cap J^-(q)$ is compact for all $p$ and $q$. 

\medskip
\medskip

\begin{prop}\label{bd J for GH prop}
Let $(M,g)$ be a globally hyperbolic Lipschitz spacetime. Then
\[
\pd J^+(p) \,=\,  J^+(p) \,\setminus\, I^+(p)\,.
\]
\end{prop}

\proof
By Corollary \ref{bd for Lipschitz cor}, it suffices to show $J^+(p)$ is closed for globally hyperbolic spacetimes. This follows from Proposition 3.5 in \cite{LingCausal}.
\qed

\medskip
\medskip

A set $S \subset M$ is a \emph{Cauchy surface} for a $C^0$ spacetime $(M,g)$ provided every inextendible causal curve intersects $S$ exactly once. From {\cite[Section 5]{Samann} we know that a $C^0$ spacetime $(M,g)$ is globally hyperbolic if and only if it has a Cauchy surface. The following result will be used in this paper.

\medskip
\medskip

\begin{thm}\label{acausal thm}
Let $(M,g)$ be a globally hyperbolic and Lipschitz spacetime. If $S \subset M$ is an acausal, compact, and $C^0$ hypersurface. Then $S$ is a Cauchy surface.
\end{thm}

\noindent \emph{Sketch of proof.}
First one shows that $J^+(S) = S \sqcup I^+(S)$. This follows because $S$ is acausal and a $C^0$ hypersurface, and the proof uses the push-up property. Next one shows that $M = I^+(S) \sqcup S \sqcup I^-(S)$. This follows by showing the right hand side is both open and closed (and thus equals $M$ since $M$ is connected). Open follows by considering a small coordinate neighborhood around a point on $S$ and using the fact that $S$ is an acausal $C^0$ hypersurface. Closed follows because $J^+(S) = S \sqcup I^+(S)$ and the fact that $J^+(S)$ is closed, which follows because $(M,g)$ is globally hyperbolic  and $S$ is compact.

Now we show $S$ is a Cauchy surface. Let $\g \colon \R \to M$ be an inextendible causal curve. Let $p = \g(0)$. By above, either $p$ lies in $I^+(S)$ or $S$ or $I^-(S)$. If $p \in S$, then we're done. Suppose $p \in I^+(S)$. Claim: there exists a $t_0 < 0$ such that $\g(t_0) \notin J^+(S)$. Suppose not. Then $\g|_{(-\infty, 0)}$ is a past inextendible causal curve contained in the compact set $J^-(p) \cap J^+(S)$ which contradicts strong causality \cite[Prop. 3.3]{LingCausal}. This proves the claim. Thus $\g(t_0) \in I^-(S)$. Since $S$ separates $M$, there is a $t_1 \in (t_0, 0)$ such that $\g(t_1) \in S$. Hence $\g$ intersects $S$. If $p \in I^-(S)$, then one applies the time-dual of the above proof. 
\qed

\medskip
\medskip

\noindent\emph{Remark.} Although it's not necessary for our paper, Theorem \ref{acausal thm} holds even when $S$ is only locally acausal.

\section{Asymptotics - a class of examples}\label{app:examples}

In this section we consider the class of examples \eqref{rcoords}
from  section \ref{Def sec}, and obtain conditions under which these examples 
 are asymptotically $\text{AdS}^2\times S^{n-1}$ in the sense of Definition \ref{maindef}.} The metric for these examples may be written as,
\begin{align*}
g \,&=\, -f(r)dt^2 + \frac{1}{f(r)}dr^2 + d\omega^2_{n-1}
\\
&=\,
f(r)\left[-dt^2 + \frac{1}{f^2(r)}dr^2 + \frac{1}{f(r)}d\omega^2_{n-1} \right]
\end{align*}
where $f(r) > 0$ is a smooth positive function.
We define a new coordinate $x$ via $r(x) = - \tan(x - \pi/2)$. The domain of $x$ is $0 < x < x_0$ where $0 < x_0 < \pi/2$ is a constant. Note that $r(x)$ is a decreasing function of $x$ and $r = \infty$ corresponds to $x = 0$.  We have $dr = -\sec^2(x-\pi/2)dx = -\csc^2(x)dx$. Letting $F(x) = f\circ r(x)$, we have
\[
g \,=\, \frac{1}{\Omega^2(x)} \underbrace{\big(-dt^2 +G^2(x) \, h\big)}_{\bar g} 
\] 
where
\begin{itemize}

\item[-] $\Omega(x) \,=\, 1/\sqrt{F(x)}$
\item[-] $G(x) \,=\, \csc^2(x)/F(x)$
\item[-] $h \,=\, dx^2 + a^2(x) d\Omega^2_{n-1}$
\item[-] $a(x) \,=\,  \sqrt{F(x)}\sin^2(x)$\,.

\end{itemize}

\medskip
\medskip

\noindent{\bf Example.} Let $f(r) = 1 + r^2$. Then $F(x) = 1 + \tan^2(x - \pi/2) = \csc^2(x)$. Therefore $\Omega(x) = \sin(x)$ and $\ov{g} = -dt^2 + dx^2 + \sin^2(x)d\omega^2_{n-1}$. Hence $\ov{g}$ is the metric for the Einstein static universe. In this case $x  = 0$ (i.e. $r = \infty$) corresponds to the north pole of $S^n$ within the Einstein static universe. From the example in the beginning of section \ref{Def sec}, we see that $g$ is the metric for $\text{AdS}^2 \times S^{n-1}$. In this example, $x = 0$ is a coordinate singularity which represents the north pole of $S^n$. In the spacetime $\ov{M}$, $x = 0$ represents the timelike line $\mathcal{J}$.

\medskip
\medskip

We want to find sufficient conditions on $f(r)$ such that  $x = 0$ (i.e. $r = \infty$) represents a coordinate singularity as in the example above. Sufficient conditions are given in Theorem \ref{app C thm}. Afterwards we show how to apply the theorem to Schwarzschild-$\text{AdS}_2\times S^{n-1}$.

\medskip

\begin{thm}\label{app C thm}
Suppose $a(x)$ is smooth and satisfies
\[
a(x) \,=\, x + O(x^{2}) \:\:\:\: \text{ and } \:\:\:\: a'(x) \,=\, 1 + O(x)
\]
Then $(M,g)$ satisfies conditions (ii), (iii), and (iv) in Definition \ref{maindef}, with $\J$ given by $x = 0$ in $\ov{M}$.
\end{thm}

\medskip

\noindent\emph{Remark.} Any scale factor $a(x)$ with a convergent Taylor expansion of the form $a(x) = x + c_2x^2 + c_3x^3 + \dotsb$ will satisfy the hypotheses of Theorem \ref{app C thm}.

\medskip

\proof
The metric $h$ is
\begin{align*}
h \,&=\, dx^2 + a^2(x)d\omega^2_{n-1}
\\
&=\, dx^2 + a^2(x)\big(d\theta^2 + \sin^2\theta d\omega^2_{n-2}\big)\,.
\end{align*}
We define new coordinates $z$ and $\rho$ given by 
\[
z(x,\theta) \,=\, b(x)\cos\theta \:\:\:\: \text{ and } \:\:\:\: \rho(x,\theta) \,=\, b(x)\sin\theta
\]
where $b(x) = e^{\int_{x_0}^x\frac{1}{a}}$, and where $x_0 > 0$ is the constant given by the domain of $x$. Note that $b' = b/a$. Therefore
\[
dz^2 + d\rho^2 \,=\, \left(\frac{b(x)}{a(x)}\right)^2dx^2 + b^2(x)d\theta^2\,.
\]
Multiplying by $(1/b')^2 = (a/b)^2$, we see that the metric $h$ in these coordinates is given by
\[
h \,=\, \frac{1}{b'(x)^2}\left(dz^2 + d\rho^2 + \rho^2d\omega^2_{n-2} \right)\,.
\]
Note that the metric in parentheses is just the Euclidean metric on $\R^n$ written in cylindrical coordinates with $\rho$ denoting the radius variable.  A simple analysis argument shows that the hypothesis $a(x) = x + O(x^2)$ implies $b(0) := \lim_{x \to 0}b(x) = 0$. Therefore $x = 0$ corresponds to the origin $z = \rho = 0$. Moreover, the same analysis used in the proof of \cite[Theorem 3.4]{LingBigBang} shows that $ 0 < b'(0) < \infty$ where $b'(0) := \lim_{x \to 0}b'(x)$. Hence $h$ does not have a degeneracy at $x = 0$ and so $x = 0$ is merely a coordinate singularity. 

To finish the proof, we have to show 
\begin{enumerate}

\item[(1)] $\Omega(0) := \lim_{x \to 0} \Omega(x) = 0$ and $d\Omega$ remains bounded on a neighborhood of $x = 0$. 

\item[(2)] $G(0) \in (0, \infty)$ where $G(0) := \lim_{x \to 0}G(x)$ and $G \circ x (z,\rho)$ extends to a Lipschitz function on a neighborhood of the origin $z = \rho = 0$.

\item[(3)] $h$ extends to a Lipschitz metric on a neighborhood of the origin $z = \rho = 0$.

\end{enumerate}
Note that (2) and (3) together imply that $G^2(x)h$ extends to a Lipschitz metric on a neighborhood of the origin; hence $\ov{g}$ extends to a Lipschitz metric on the timelike line $x = 0$.

We first show (1).  $\Omega(0) = 0$ follows from a simple analysis argument using the fact that $\Omega(x) = \sin^2(x)/a(x)$ and the hypothesis $a(x) = x + O(x^{2})$. Now we show $d\Omega$ remains bounded on a neighborhood of $x = 0$. Note that $b$ is a strictly increasing smooth function which is never zero. Therefore it is invertible and the derivative of its inverse is $(b^{-1})'\big(b(x)\big) = 1/b'(x)$. Recall that $x = b^{-1}\big(\sqrt{z^2 + \rho^2} \big)$, so $\pd x/\pd z = z/(b'b) = \cos(\theta)/b'$ and $\pd \rho /\pd z = \rho /(b'b) = \sin(\theta) / b' $. Both are bounded near $(z,\rho)=(0,0)$, hence $dx=\frac{1}{b'}(\cos(\theta) dz +\sin(\theta) d\rho )$ is bounded as well. So boundedness of $d\Omega$ follows from the limit
\[
\Omega'(x) \,=\, \frac{2\sin(x)\cos(x)a(x) - a'(x) \sin^2(x)}{a^2(x)} \:\to\: 1 \:\:\:\: \text{ as } \:\:\:\: x \to 0\,.
\] 

Now we show (2). In fact we have $G(0) = 1$. This follows because $a(x) = \sin(x)/\sqrt{G(x)}$ and the hypothesis $a(x) = x + O(x^2)$ along with an application of the squeeze theorem. Now we show $G \circ x(\rho,z)$ extends to a Lipschitz function on a neighborhood of the origin $z = \rho = 0$.
From elementary analysis, it suffices to show that the limits
\[
\lim_{(z,\rho) \to (0,0)}\frac{\pd(G \circ x)}{\pd z} \:\:\:\: \text{ and } \:\:\:\: \lim_{(z, \rho) \to (0,0)} \frac{\pd(G \circ x)}{\pd \rho}
\]  
remain bounded. Since $G(x) = \sin^2(x)/a^2(x)$, the chain rule gives
\[
\frac{\pd G}{\pd z} \,=\, G'(x)\frac{\pd x}{\pd z} \,=\, \left(\frac{2\sin(x)\cos(x)}{a^2} - \frac{2a'\sin^2(x)}{a^3}\right)\left(\frac{z}{bb'}\right)\,.
\]
Using $b' = b/a$ and $z=b\cos(\theta)$, we get
\[
\frac{\pd (G \circ x)}{\pd z} \,=\, \left(\frac{2\sin(x)\cos(x)}{ab} - \frac{2a' \sin^2(x)}{a^2b} \right)\cos(\theta)\,.
\]
{Analysis analogous to the proof of \cite[Theorem 3.4]{LingBigBang} shows that $b(x) = x/x_0 + O(x^2)$. This combined with the hypotheses on $a(x)$ shows that the term in larger brackets above remains bounded as $x \to 0$.  Thus $\pd G/ \pd z$ remains bounded as $(z, \rho) \to (0,0)$.  Similarly, the same result holds for $\pd G/ \pd \rho$. 

Now we show (3). A similar argument as used in the proof of (2) shows that for $\omega(x) = 1/b'(x)$, we have that the limits of $\pd \omega/\pd z$ and $\pd \omega/ \pd \rho$ remain bounded as $(z, \rho) \to (0,0)$. Hence $\omega \circ x(\rho,z)$ will be Lipschitz on a neighborhood of the origin. 
\qed

\medskip
\medskip

\noindent{\bf Example.} Let $f(r) = 1 + r^2 - 2m/r$ which corresponds to Schwarzschild-$\text{AdS}_2\times S^{n-1}$ 
(see equation (\ref{SSAdS_2})).
We will show this $f(r)$ satisfies the conditions of Theorem \ref{app C thm}. Since $r(x) = - \tan(x - \pi/2) = -\cos(x)/\sin(x)$, we have 
\[F(x) \,=\, f \circ r(x) \,=\, 1 + \frac{\cos^2(x)}{\sin^2(x)} + 2m\tan(x)\,=\, \frac{1}{\sin^2(x)} +2m \tan(x)\,.
\]
Therefore $F(x) \sin^4(x) = \sin^2(x) + 2m \tan(x) \sin^4(x)$. Hence 
\[
a(x) \,=\, \sin(x)\sqrt{1 + 2m \tan(x)\sin^2(x)} \,=\, x - \frac{x^3}{6} +  m x^4 + \dotsb
\] Therefore $a(x)$ satisfies the hypotheses of Theorem \ref{app C thm} (see the remark after the theorem).

\medskip
\medskip

Now we give general conditions on $f(r)$ in equation (\ref{SSAdS_2})  to satisfy the assumptions for $a(x)$ in Theorem \ref{app C thm}.

\medskip
\medskip

\begin{cor}If $f$ satisfies $f(r)=r^2+O(r)$ and $f'(r)=2r+O(1)$ as $r\to\infty$, then $a(x)$ satisfies $a(x)=x+O(x^2)$ and $a'(x)=1+O(x)$. If further $f:(-\infty,\infty)\to (0,\infty)$ satisfies these asymptotics as both $r\to \infty$ and $r\to -\infty$, then $\R\times (-\infty,\infty)\times S^{n-1}$ with metric
	\begin{align*}
	g \,&=\, -f(r)dt^2 + \frac{1}{f(r)}dr^2 + d\omega^2_{n-1}
	\end{align*}
	has two asymptotically $AdS_2\times S^{n-1}$ ends.
\end{cor}
\begin{proof}
	We first show that $a(x)=$ satisfies $a(x)=x+O(x^2)$ and $a'(x)=1+O(x)$. We have $a(x)=\sqrt{F(x)}\sin^2(x)$ where $F(x)=f(-\tan(x-\pi/2))$. Since $-\tan(x-\pi/2)=1/x+O(x)$, we get $F(x)=1/x^2 +O(1/x)$ from $f(r)=r^2+O(r)$. Hence, $\sqrt{F(x)}=1/x+O(1)$ and $a(x)=x+O(x^2)$ follows.
	
	For $a'$, note that \[a'(x)=\frac{F'(x)}{2\sqrt{F(x)}} \sin^2(x)+2\sqrt{F(x)} \sin(x)\cos(x).\] Using that $\sqrt{F(x)}=1/x+O(1)$, we immediately get that the second summand is $2+O(1)$. For the first summand, we need to work out $F'(x)$. We have $F'(x)=f'(r(x))r'(x)=-f'(r(x)) 1/\sin^2(x)$. Since $f'(r)=2r+O(1)$, we get $f'(r(x))=2/x+O(x)+O(1)=2/x+O(1)$ and hence  $F'(x)=-2/x^3+O(1/x^2)$. Using this, we obtain that the first summand in the expression for $a'$ is $-1+O(x)$, hence $a'(x)=1+O(x)$.
	
	If $f$ satisfies the same asymptotics as $r\to -\infty$ we clearly get similar asymptotics for $a$ as $x\to \pi$ using the same change of coordinates: $a(x)=\pi-x+O((\pi-x)^2)$ and $a'(x)=1+O((\pi-x))$. So, as in the proof of Theorem \ref{app C thm}, we get that $\bar{g}$ also extends to $x=\pi$, so we get a conformal Lipschitz extension to all of $\R\times S^3=:\ov{M}$. Since $\bar{g}=-dt^2+G^2(x)h$, any hypersurface of the form $\{t_0\}\times S^3$ is a Cauchy surface, so $(\ov{M},\ov{g})$ is globally hyperbolic.
\end{proof}

\medskip
\medskip

\noindent{\bf Comparison with the asymptotics in \cite{GalGraf}.} Let $M=\R\times (a,\infty) \times S^2$ with metric $g=\mathring{g}+h$, where $\mathring{g}=g_{AdS_2\times S^2}=-\cosh^2(\sigma) dt^2 +d\sigma^2 +d\omega^2$ and $h$ decays as in the definition of asymptotically $AdS_2\times S^2$ in \cite{GalGraf}. This in particular means that  we have
\begin{align*}
	 h(e_i,e_j) =O(1/\sigma)
\end{align*}
 for any $\mathring{g}$-othonormal basis $\{e_i(p)\}_{i=0}^3$ with $e_0 = \frac1{\cosh \sigma}\frac{\partial}{\partial t}$.

We now show that, similarly to the examples discussed above, we can interpret $x\to \infty$ as an ``almost'' asymptotically $AdS_2\times S^2$ end: Defining $x(\sigma) $ via $\sinh(\sigma)=r=-\tan(x-\pi/2)$ we get $M=\R \times (0, x(a)) \times S^2$ and, as above, $\mathring{g}$ becomes \[\mathring{g}=-(1+r^2)dt^2+\frac{1}{1+r^2}dr^2+d\omega^2=\frac{1}{\sin^2(x)} g_{\R\times S^3},\]
so we define 
	
	\[
	\bar{g}:=
	\sin^2(x) g=g_{\R\times S^3}+\sin^2(x) h \quad \mathrm{on}\,\; M.
	\]
To show that this continuously extends to $x=0$ we show that $\sin^2(x)h\to 0$ as $x\to 0$. To see this, let $\bar{e}_0:=\partial_t$ and let $\{\bar{e}_i\}_{i=1}^3$ be an orthonormal frame for the round $S^3$ near the north pole. Then $e_i:=\Omega \bar{e}_i$, $i=0,\dots,3$ is a $\mathring{g}$-orthonormal frame with $e_0=1/\cos(\sigma) \partial_t$, so for any $i,j$
\[\sin^2(x)h(\bar{e}_i,\bar{e_j})=h(e_i,e_j)=O(\frac{1}{\sigma})\to 0 \]
as $x\to 0$.
	
	Since the conformal factor $\Omega$ we used here is the same as for exact $AdS_2\times S^2$, we immediately get that $d\Omega$ remains bounded as $x\to \infty$.
	
	So, except for $\bar{g}$ possibly being merely continuous and not Lipschitz, $(M,g)$ satisfies (ii)--(iv) from Definition \ref{maindef}!

\medskip

\noindent
{\it Remark.} Regarding Lipschitz continuity of $\bar{g}$ we observe the following: The asymptotics in \cite{GalGraf} stipulate that \[e_k(h(e_i,e_j))=O(1/\sigma).\]
Trying to estimate $\bar{e}_k(\Omega^2 h(\bar{e}_i,\bar{e}_j))$ using this yields
\begin{align*}
|\bar{e}_k(\Omega^2 h(\bar{e}_i,\bar{e}_j))|&=  \frac{1}{\Omega} |e_k(h(e_i,e_j))| \leq C \frac{1}{\sigma\,\sin(x)} \\
&= C \frac{1}{\sin(x)\sinh^{-1}(-\tan(x-\pi/2))}\to \infty \quad \mathrm{as}\quad x \to 0.
\end{align*}
So the asymptotics in \cite{GalGraf} are not sufficient to get Lipschitz continuity of the extension. Note however that replacing $e_k(h(e_i,e_j))=O(1/\sigma )$ with the stronger assumption $e_k(h(e_i,e_j))=O(1/\exp(\sigma ))$ would imply boundedness of the derivatives estimated above, i.e., Lipschitzness of $\bar{g}$.

\medskip

\noindent
\textsc{Acknowledgements.}  GJG and MG would like to thank Paul Tod for previous communications in connection with reference \cite{GalGraf}.   The research of GJG was partially supported by the NSF under the grant DMS-171080. 
Part of the work on this paper was supported by the Swedish Research Council under grant no. 2016-06596 while GJG and EL were participants at Institut Mittag-Leffler in Djursholm, Sweden during the Fall semester of 2019. Parts of this work were carried out while MG was at the University of T\"{u}bingen.

\medskip
\medskip

\bibliographystyle{amsplain}
\bibliography{ads2b}

\end{document}